\newif\ifprocs
\procsfalse

\ifprocs
\documentclass[oribibl]{llncs}
\usepackage{amsmath,amssymb,amsfonts,latexsym}
\usepackage[english]{babel}
\else
\documentclass[onecolumn,11pt,a4paper]{article}
\usepackage[margin=0.5in]{geometry}
\usepackage[english]{babel}
\usepackage{amsmath,amssymb,amsfonts,latexsym}
\usepackage{amsthm}
\usepackage[fleqn,tbtags]{mathtools}
\usepackage{lastpage}
\usepackage{a4wide}
\usepackage{fancyhdr}
\fi

\usepackage[boxed]{algorithm}
\usepackage[noend]{algorithmic}
\usepackage{graphicx}
\usepackage{color}
\usepackage{xspace}
\usepackage{enumerate}

\usepackage{ifpdf}
\ifpdf    
\usepackage{hyperref}
\else    
\usepackage[hypertex]{hyperref}
\fi

\ifprocs
\else
\fi

\ifprocs
\else
\makeatletter
\newtheorem*{rep@theorem}{\rep@title} \newcommand{\newreptheorem}[2]{%
\newenvironment{rep#1}[1]{%
\def\rep@title{\bf #2 \ref{##1} }%
\begin{rep@theorem} }%
{\end{rep@theorem} } }
\makeatother

\newreptheorem{lemma}{Lemma}
\newreptheorem{corollary}{Corollary}
\fi

\ifprocs
\let\oldproof\proof%
\let\endoldproof\endproof%
\renewenvironment{proof}
  {\oldproof}{\hfill\squareforqed\endoldproof}
\else
\newtheorem{theorem}{Theorem}[section]
\newreptheorem{theorem}{Theorem}

\newtheorem{lemma}[theorem]{Lemma}

\newtheorem{definition}{Definition}[section]

\newtheorem{claim}[theorem]{Claim}

\theoremstyle{definition}
\newtheorem{update}{Bibliographic Update}[section]

\fi

\def\compactify{\itemsep=0pt \topsep=0pt \partopsep=0pt \parsep=0pt}
\newcommand {\ignore} [1] {}

\DeclareMathOperator{\pa}{\texttt{Par}}
\DeclareMathOperator{\cut}{\texttt{Cut}}
\DeclareMathOperator{\mincut}{\texttt{mincut}}
\DeclareMathOperator{\distinct}{}

\DeclareMathOperator{\dash}{-}
\DeclareMathOperator{\redundancy}{redundancy}
\providecommand{\eqdef}{:=}
\providecommand{\aset}[1]{\{#1\}}
\providecommand{\card}[1]{\lvert#1\rvert}

\newcommand{\gcut}{{\sc Group-Cut}\xspace}
\newcommand{\mwaycut}{{\sc Multiway-Cut}\xspace}
\newcommand{\mcut}{{\sc Multicut}\xspace}

\newcommand{\gcutk}{{\sc $(\alpha,\beta)$-Group-Cut}\xspace}

\newcommand{\gcutab}{{\sc $(\alpha,\beta)$-Group-Cut}\xspace}
\newcommand{\gcutlower}{{\sc $(2,1)$-Group-Cut}\xspace}
\newcommand{\mwaycutk}{{\sc $k$-Multiway-Cut}\xspace}
\newcommand{\mcutk}{{\sc $k$-Multicut}\xspace}
\def\ala{\`a la\xspace}
\def\Benczur{Bencz\'{u}r\xspace}
\newcommand{\etal}{{\em et al.\ }\xspace}

\def \calQ {{\cal Q}}
\def \calM {{\cal M}}

\definecolor{orange}{rgb}{1,0.65,0}

\newcommand{\one}{v_0}

\ifprocs
\title{Tight Bounds for Gomory-Hu-like \\ Cut Counting%
\thanks{This work was supported in part by the Israel Science Foundation (grant
\#897/13).}
}
\author{Rajesh Chitnis
\and Lior Kamma
\and Robert Krauthgamer
}
\institute{Weizmann Institute of Science, Rehovot, Israel. \email{\{rajesh.chitnis,lior.kamma,robert.krauthgamer\}@weizmann.ac.il}
}

\else
\title{Tight Bounds for Gomory-Hu-like Cut Counting%
\thanks{This version contains additional references to previous work
(which have some overlap with our results),
see Bibliographic Update~\ref{hassin-update}.}
}
\author{Rajesh Chitnis%
\thanks{Supported by the I-CORE Program of the Planning and Budgeting Committee and The Israel Science Foundation (grant No. 4/11).}
\qquad Lior Kamma%
\thanks{This work was supported in part by the Israel Science Foundation (grant
\#897/13).}
\qquad Robert Krauthgamer\footnotemark[\value{footnote}]
\\
Weizmann Institute of Science
\\
\texttt{\small \{rajesh.chitnis,lior.kamma,robert.krauthgamer\}@weizmann.ac.il}
}
\fi

\begin{document}

\maketitle

\begin{abstract}
By a classical result of Gomory and Hu (1961),
in every edge-weighted graph $G=(V,E,w)$,
the minimum $st$-cut values, when ranging over all $s,t\in V$,
take at most $|V|-1$ distinct values.
That is, these $\binom{|V|}{2}$ instances exhibit
\emph{redundancy factor} $\Omega(|V|)$.
They further showed how to construct from $G$
a tree $(V,E',w')$ that stores all minimum $st$-cut values.
Motivated by this result, we obtain \emph{tight} bounds for the redundancy factor of several generalizations of the minimum $st$-cut problem.

\begin{enumerate} \compactify
\item
\gcut: Consider the minimum $(A,B)$-cut, ranging over
all subsets $A,B\subseteq V$ of given sizes $|A|=\alpha$ and $|B|=\beta$.
The redundancy factor is $\Omega_{\alpha,\beta}(|V|)$.
\item
\mwaycut: Consider the minimum cut separating every two vertices
of $S\subseteq V$, ranging over all subsets of a given size $|S|=k$.
The redundancy factor is $\Omega_{k}(|V|)$.
\item
\mcut: Consider the minimum cut separating every demand-pair
in $D\subseteq V\times V$, ranging over collections of $|D|=k$ demand pairs.
The redundancy factor is $\Omega_{k}(|V|^k)$.
This result is a bit surprising, as the redundancy factor is much larger
than in the first two problems.
\end{enumerate}

A natural application of these bounds is to construct small data structures that stores all relevant cut values, \ala the Gomory-Hu tree.
We initiate this direction by giving some upper and lower bounds.
\end{abstract}

\section{Introduction}

One of the most fundamental combinatorial optimization problems
is \emph{minimum $st$-cut},
where given an edge-weighted graph $G=(V,E,w)$ and two vertices $s,t\in V$,
the goal is to find a set of edges of minimum total weight that separates $s,t$
(meaning that removing these edges from $G$ ensures there is no $s\dash t$ path).
This problem was studied extensively,
see e.g.~the famous minimum-cut/maximum-flow duality~\cite{ford1956maximal},
and can be solved in polynomial time.
It has numerous theoretical applications,
such as bipartite matching and edge-disjoint paths,
in addition to being extremely useful in many practical settings,
including network connectivity, network reliability, and image segmentation,
see e.g.~\cite{AMO93} for details.
Several generalizations of the problem, such as multiway cut, multicut, and $k$-cut, have been well-studied in operations research and theoretical computer science.

In every graph $G=(V,E,w)$, there are in total $\binom{|V|}{2}$ instances of the minimum $st$-cut problem, given by all pairs $s,t\in V$.
Potentially, each of these instances could have a different value for the minimum cut.
However, the seminal work of Gomory and Hu~\cite{GH61} discovered
that \emph{undirected} graphs admit a significantly stronger bound
(see also~\cite[Lemma 8.15]{AMO93} or \cite[Section 3.5.2]{CCPS98}).
\begin{theorem}[\cite{GH61}]
\label{thm:GH}
Let $G=(V,E,w)$ be an edge-weighted undirected graph.
Then the number of distinct values over all possible $\binom{|V|}{2}$ instances of the minimum $st$-cut problem is at most $|V|-1$.
\end{theorem}

The beautiful argument of Gomory and Hu shows the existence of
a tree $\mathcal{T}=(V,E',w')$, usually called a \emph{flow-equivalent tree},
such that for every $s,t\in V$ the minimum $st$-cut value in $\mathcal T$
is exactly the same as in $G$.
(They further show how to construct a so-called \emph{cut-equivalent tree},
which has the stronger property that every vertex-partitioning that attains
a minimum $st$-cut in $\mathcal T$, also attains a minimum $st$-cut in $G$;
see Section~\ref{sec:related} for more details on this and related work.)
Every $G$ which is a tree (e.g., a path) with distinct edge weights has
exactly $|V|-1$ distinct values,
and hence the Gomory-Hu bound is existentially tight.

Another way to state Theorem~\ref{thm:GH} is that there is always a huge redundancy between the $\binom{|V|}{2}$ minimum $st$-cut instances in a graph.
More precisely, the ``redundancy factor'',
measured as the ratio between the number of instances and
the number of distinct optimal values attained by them,
is always $\Omega(|V|)$.
We study this question of redundancy factor for the following generalizations of minimum $st$-cut.
Let $G=(V,E,w)$ be an undirected edge-weighted graph.
\begin{itemize} \compactify
\item \textbf{\gcut}:
Given two disjoint sets $A,B\subseteq V$ find a minimum $(A,B)$-cut,
i.e., a set of edges of minimum weight that separates every vertex in $A$
from every vertex in $B$.

\item \textbf{\mwaycut}: Given $S\subseteq V$ find a minimum-weight set of edges, whose removal ensures that for every $s\neq s'\in S$
there is no $s \dash s'$ path.

\item \textbf{\mcut}: Given $Q\subseteq V\times V$ find a minimum-weight set of edges, whose removal ensures that for every $(q,q')\in Q$ there is no $q \dash q'$ path.
\end{itemize}

In order to present our results about the redundancy in these cut problems
in a streamlined way,
we introduce next the terminology of vertex partitions and demand graphs.

\paragraph{Cut Problems via Demand Graphs.}
Denote by $\pa(V)$ the set of all partitions of $V$,
where a {\em partition} of $V$ is, as usual,
a collection of pairwise disjoint subsets of $V$ whose union is $V$.
Given a partition $\Pi \in \pa(V)$ and a vertex $v \in V$,
denote by $\Pi(v)$ the unique $S \in \Pi$ satisfying $v \in S$.
Given a graph $G=(V,E,w)$,
define the function $\cut_G:\pa(V)\rightarrow \mathbb{R}^{\geq 0}$
to be $\cut_G(\Pi)= \sum_{uv\in E\,:\,\Pi(u)\neq \Pi(v)} w(uv)$.
We shall usually omit the subscript $G$,
since the graph will be fixed and clear from the context.

Cut problems as above can be defined by specifying the graph $G$ and
a collection $D$ of \emph{demands},
which are the vertex pairs that need to be separated.
We can view $(V,D)$ as an (undirected and unweighted) \emph{demand graph},
and by slight abuse of notation, $D$ will denote both this graph and its edges.
For example, an instance of \gcut is defined by $G$ and demands that form
a complete bipartite graph $K_{A,B}$ (to formally view it as a graph on $V$,
let us add that vertices outside of $A\cup B$ are isolated).
We say that partition $\Pi \in \pa(V)$ \emph{agrees} with $D$
if every $uv \in D$ satisfies $\Pi(u) \ne \Pi(v)$.
The optimal cut-value for the instance defined by $G$ and $D$
is given by
$$
  \mincut_G(D)
  \eqdef \min \aset{\cut_G(\Pi) : \text{$\Pi \in \pa(V)$ agrees with $D$} }.
$$

\paragraph{Redundancy among Multiple Instances.}
We study multiple instances on the same graph $G=(V,E,w)$
by considering a family $\mathcal D$ of demand graphs.
For example, all minimum $st$-cut instances in a single $G$
corresponds to the family $\mathcal D$ of all demands of the form $D=\aset{(s,t)}$ (i.e., demand graph with one edge).
The collection of optimal cut-values over the entire family $\mathcal D$
of instances in a single graph $G$,
is simply $\aset{ \mincut(D): D\in \mathcal{D} }$.
We are interested in the ratio between the size of this collection
as a multiset and its size as a set,
i.e., with and without counting multiplicities.
Equivalently, we define the \emph{redundancy factor}
of a family $\mathcal D$ of demand graphs to be
$$
  \redundancy(\mathcal D)
  \eqdef \frac{\card{\mathcal D}} {\card{\aset{\mincut(D): D\in \mathcal D}}} \ ,
$$
where throughout, $\card{A}$ denotes the size of $A$ as a \emph{set},
i.e., ignoring multiplicities.

\paragraph{Motivation and Potential Applications.}
A natural application of the redundancy factor is to construct small data structures that stores all relevant cut values.
For the minimum $st$-cut problem,
Gomory and Hu were able to collect all the cut values into a tree on the same vertex set $V$.
This tree can easily support fast query time,
or a distributed implementation (labeling scheme) \cite{KKKP05}.

In addition, large redundancy implies that there is a small collection of cuts
that contains a minimum cut for each demand graph.
Indeed, first make sure all cut values in $G$ are distinct
(e.g., break ties consistently by perturbing edge weights),
and then pick for each cut-value in $\aset{\mincut(D): D\in \mathcal D}$
just one cut that realizes it.
This yields a data structure that reports,
given demands $D\in\mathcal D$,
a vertex partition that forms a minimum cut (see more in Section~\ref{sec:extensions}).

\subsection{Main Results}
\label{subsection:main-results}

Throughout, we denote $n=|V|$.
We use the notation $O_{\gamma}(\cdot)$ to suppress factors that depend
only on $\gamma$, and similarly for $\Omega$ and $\Theta$.

\paragraph{The \gcut problem.}
In this problem, the demand graph is a complete bipartite graph $K_{A,B}$
for some subsets $A,B\subset V$.
We give a tight bound on the redundancy factor of the family of all instances
where $A$ and $B$ are of given sizes $\alpha$ and $\beta$, respectively.
The special case $\alpha=\beta=1$ is just all minimum $st$-cuts in $G$,
and thus recovers the Gomory-Hu bound (Theorem~\ref{thm:GH}).

\begin{theorem}
\label{thm:group}
For every graph $G=(V,E,w)$ and $\alpha,\beta\in \mathbb{N}$, we have
$\card{\aset{ \mincut(K_{A,B}): |A|=\alpha, |B|=\beta }}
 = O_{\alpha,\beta}(n^{\alpha+\beta-1})$,
hence the family of $(\alpha,\beta)$-group-cuts
has redundancy factor $\Omega_{\alpha,\beta}(n)$.
Furthermore, this bound is existentially tight (attained by some graph $G$) for all $\alpha$, $\beta$ and $n$.
\end{theorem}

\paragraph{The \mwaycut problem.}
In this problem, the demand graph is a complete graph $K_S$
for some subset $S\subseteq V$.
We give a tight bound on the redundancy factor of the family of all instances
where $S$ is of a given size $k\ge 2$.
Again, the Gomory-Hu bound is recovered by the special case $k=2$.

\begin{theorem}
\label{thm-multiway}
For every graph $G=(V,E,w)$ and for every integer $k\in \mathbb{N}$, we have
$\card{\aset{ \mincut(K_{S}): |S|=k }} = O_{k}(n^{k-1})$,
hence the family of $k$-multiway-cuts has redundancy factor $\Omega_{k}(n)$.
Furthermore, this bound is existentially tight for all $n$ and $k$.
\end{theorem}

\paragraph{The \mcut problem.}
In this problem, the demand graph is a collection $D$ of demand pairs.
We give a tight bound on the redundancy factor of the family of all instances
where $D$ is of a given size $k\in \mathbb{N}$.
Again, the Gomory-Hu bound is recovered by the special case $k=1$.

\begin{theorem}
\label{thm-multicut}
For every graph $G=(V,E,w)$ and $k\in \mathbb{N}$, we have
$\card{\aset{ \mincut(D) : D\subseteq V\times V,\ \card{D}=k }}
  = O_{k}(n^{k})$,
and hence the family of $k$-multicuts has redundancy factor $\Omega_{k}(n^k)$.
Furthermore, this bound is existentially tight for all $n$ and $k$.
\end{theorem}

Theorem~\ref{thm-multicut} is a bit surprising, since it shows
a redundancy factor that is polynomial, rather than linear, in $n$
(for fixed $\alpha,\beta$ and $k$),
so in general \mcut has significantly larger redundancy than \gcut and \mwaycut.

\begin{update}
\label{hassin-update}
Refael Hassin brought to our attention [Private communication, November 2017]
prior work that has overlap with some of our results. 
Hassin~\cite{hassin-88} had previously obtained our upper bounds for \mwaycut (Theorem~\ref{thm-multiway}) and \mcut (Theorem~\ref{thm-multicut}),
using the same ``matrix'' proof technique that we use in Section~\ref{subsection:group-upper-matrix}).
In another paper, Hassin~\cite{hassin1990algorithm} showed how to efficiently compute all the distinct values for various problems such as \mwaycut and \mcut. More precisely, if a problem can have at most $X$ distinct values (in worst-case), then his algorithm computes them by solving only $O(X)$ many instances of that problem.
For the \mwaycut and \mcut problems, Hartvigsen~\cite{hartvigsen} showed how to compute a matrix of size $O(X)$ (again $X$ is the maximum number of distinct values in worst-case) such that the solution for any given instance can be obtained in $X^{O(1)}$ time. Additionally, non-trivial redundacy has been shown for other problems such as {\sc xcut} 
(which asks for a minimum cut such that given vertices
$s,t\in V$ are on the \emph{same} side)~\cite{hassin91}, as well as other problems related to graph coloring, SAT, etc.~\cite{hassin-einstein}. 
In fact, Einstein and Hassin~\cite{hassin-einstein} also obtain the upper bounds for \mwaycut and \mcut, and use both the ``matrix'' (see Section~\ref{subsec:multiway-upper}) and the ``polynomial'' (see Section~\ref{subsec:multicut-upper}) proof techniques.
\end{update}

\subsection{Extensions and Applications}
\label{sec:extensions}

Our main results above 
actually apply more generally and have algorithmic consequences,
as discussed below briefly.

\paragraph{Terminals Version.}
In this version, the vertices to be separated
are limited to a subset $T\subseteq V$ called \emph{terminals},
i.e., we consider only demands inside $T\times T$.
All our results above
(Theorems~\ref{thm:group}, \ref{thm-multiway}, and \ref{thm-multicut})
immediately extend to this version of the problem
--- we simply need to replace $|V|$ by $|T|$ in all the bounds.
As an illustration, the terminals version of Theorem~\ref{thm:GH} states
that 
the $\binom{|T|}{2}$ minimum $st$-cuts (taken over all $s,t\in T$)
attain at most $|T|-1$ distinct values.
(See also~\cite[Section 3.5.2]{CCPS98} for this same version.)
Extending our proofs to the terminals version is straightforward;
for example, in Section~\ref{subsection:group-upper-poly} we need to consider polynomials in $|T|$ variables instead of $|V|$ variables.

\paragraph{Data Structures.}

Flow-equivalent or cut-equivalent trees,
such as those constructed by Gomory and Hu~\cite{GH61}, may be viewed
more generally as succinct data structures that support certain queries,
either for the value of an optimal cut, or for its vertex-partition,
respectively.
Motivated by this view, we define data structures,
which we call as evaluation schemes,
that preprocess an input graph $G$, a set of terminals $T$,
and a collection of demand graphs $\mathcal D$,
so as to answer a cut query given by a demand graph $D\in D$.
The scheme has two flavors, one reports the minimum cut-value,
the second reports a corresponding vertex-partition.
In Section~\ref{sec:dataStructures} we initiate the study of such schemes,
and provide constructions and lower bounds for some special cases.

\paragraph{Functions Different From Cuts.}
Recall that the value of the minimum $st$-cut equals
$\min \aset{ \cut_G(X,V\setminus X) : X\subseteq V,\ s\in X,\ t\notin X }$.
Cheng and Hu~\cite{CH91} extended the Gomory-Hu bound (Theorem~\ref{thm:GH})
to a wider class of problems as follows.
Instead of a graph $G$,
fix a ground set $V$ and a function $f: 2^V \to \mathbb{R}$.
Now for every $s,t\in V$, consider the optimal value
$\min \aset{ f(X) : X\subseteq V,\; |X \cap \{s,t\}| = 1 }$.
They showed that ranging over all $s,t\in V$,
the number of distinct optimal values is also at most $|V|-1$.
All our results above
(Theorems~\ref{thm:group}, \ref{thm-multiway}, and \ref{thm-multicut})
actually extend to every function $f: \pa(V)\rightarrow \mathbb{R}$. However, to keep the notation simple,
we opted to present all our results only for the function $\cut$.

\paragraph{Directed Graphs.} What happens if we ask the same questions for the directed variants of the three problems considered previously?
Here, an $s\rightarrow t$ cut means a set of edges whose removal ensures that no $s\rightarrow t$ path exists.
Under this definition, we can construct explicit examples for the directed variants of our three problems above 
where there is no \emph{non-trivial redundancy},
i.e., the number of distinct cut values is asymptotically equal to the total number of instances. See Appendix~\ref{app:directed-graphs} for more details.

\subsection{Related Work}
\label{sec:related}

Gomory and Hu~\cite{GH61} showed how to compute a cut-equivalent tree,
and in particular a flow-equivalent tree,
using $|V|-1$ minimum $st$-cut computations on graphs no larger than $G$.
Gusfield~\cite{Gusfield90} has shown a version where all the cut computations
are performed on $G$ itself (avoiding contractions).
For unweighted graphs, a faster (randomized) algorithm for computing a Gomory-Hu tree which runs in $\tilde{O}(|E|\cdot |V|)$ time was recently given by Bhalgat et al.~\cite{BHKP07}.

We already mentioned that Cheng and Hu~\cite{CH91} extended Theorem~\ref{thm:GH} from cuts to an arbitrary function $f: 2^V \to \mathbb{R}$.
They further showed how to construct a flow-equivalent tree for this case
(but not a cut-equivalent tree).
\Benczur~\cite{benczur} showed a function $f$ for which there is no
cut-equivalent tree. In addition, he showed that for directed graphs,
even flow-equivalent trees do not exist in general.

Another relevant notion here is that of mimicking networks,
introduced by Hagerup, Katajainen, Nishimura, and Ragde~\cite{KNR98}.
A mimicking network for $G=(V,E,w)$ and a terminals set $T\subseteq V$
is a graph $G'=(V',E',w')$ where $T\subset V'$ and for every $X,Y\in T$,
the minimum $(X,Y)$-cut in $G$ and in $G'$ have the exact same value.
They showed that every graph has a mimicking network with at most $2^{2^{|T|}}$ vertices.
Some improved bounds are known,
e.g., for graphs that are planar or have bounded treewidth,
as well as some lower bounds \cite{CSWZ00,KR13,KR14}.
Mimicking networks deal with the \gcut problem for all $A,B\subset V$;
we consider $A,B$ of bounded size, and thus typically achieve much smaller bounds.

\section{\gcut: The Case of Complete Bipartite Demands} \label{sec:groupCut}

This section is devoted to proving Theorem~\ref{thm:group}. First we give two proofs, one in Section~\ref{subsection:group-upper-poly} via polynomials and the second in Section~\ref{subsection:group-upper-matrix} via matrices, for the bound $|\distinct\{ \mincut(K_{A,B}): |A|=\alpha, |B|=\beta\}|= O_{\alpha,\beta}(n^{\alpha+\beta-1})$. Then in Section~\ref{subsection:group-lower} we construct examples of graphs for which this bound is tight. Since $|\{ K_{A,B}: |A|=\alpha, |B|=\beta\}|= \binom{n}{\alpha}\cdot \binom{n-\alpha}{\beta}= \Theta_{\alpha,\beta}(n^{\alpha+\beta})$, it follows that the redundancy factor is $\Omega_{\alpha,\beta}(n)$.

\subsection{Proof via Polynomials}
\label{subsection:group-upper-poly}

In this section we show the bound $|\distinct\{ \mincut(K_{A,B}): |A|=\alpha, |B|=\beta\}|= O_{\alpha,\beta}(n^{\alpha+\beta-1})$ using polynomials. Let $r = {{n \choose \alpha}{n-\alpha \choose \beta}}$ and let $ \{ K_{A_1, B_1}, K_{A_2,B_2}, \ldots, K_{A_r, B_r} \}$ be the set of demand graphs for \gcutk. For every vertex $v \in V$ we assign a boolean variable denoted by $\phi_v$. Given an instance $A,B$ we can assume that the optimal partition only contains two parts, one which contains $A$ and other which contains $B$, since we can merge other parts into either of these parts.

Fix some $j\in [r]$. Recall that $\Pi=\{U, V \setminus U\} \in \pa(V)$ agrees with, i.e., is a feasible solution for, the demand graph $K_{A_j,B_j}$ if and only if the following holds: $\Pi(u)\neq \Pi(v)$ whenever $u \in A_j$ and $v \in B_j$ or vice versa.

Fix arbitrary $a_j \in A_j$ and $b_j \in B_j$. We associate with the demand graph $K_{A_j, B_j}$ the formal polynomial $P_j$ over the variables $\{\phi_v: v\in V\}$
$$ P_j =\Pi_{b \in B_j}\Big(\phi_{a_j}-\phi_{b}\Big)\cdot \Pi_{a \in A_{j}\setminus \{a_j\}}\Big(\phi_a-\phi_{b_j}\Big) \;.$$
Note that $P_j$ is a polynomial of degree $\alpha+\beta-1$. Given $U \subseteq V$, we may think of $\Pi = \{U,V \setminus U\}$ as a vector in $\{0,1\}^{n}$. We denote by $P_{j}(\Pi)$ the value of the polynomial $P_j$ (over $\mathbb{F}_2$) when instantiated on $\Pi$.

\begin{lemma}
A partition $\Pi$ is feasible for the demand graph $K_{A_j, B_j}$ if and only if $P_{j}(\Pi)\neq 0$
\label{lem:feasibleNonZero}
\end{lemma}
\begin{proof}
Suppose $\Pi$ is feasible for the demand graph $K_{A_j,B_j}$. So $\Pi(u)\neq \Pi(v)$ if $u\in A_j, v\in B_j$ or vice versa. Since every term of $P_{j}$ contains one variable from each of $A_j$ and $B_j$, it follows that $P_{j}(\Pi)\neq 0$.

Conversely, assume $P_{j}(\Pi)\neq 0$.
Let $u \in A_j$. Since $\Pi(u) \ne \Pi(b_j)$ and $\Pi(b_j) \ne \Pi(a_j)$ it follows that $\Pi(u)=\Pi(a_j)$. Similarly for every $v \in B_j$, $\Pi(v)=\Pi(b_j)$. Therefore, it follows that $\Pi(u)\neq \Pi(v)$ whenever $u \in A_j$ and $v \in B_j$ or vice versa, i.e., $\Pi$ is feasible for $K_{A_j,B_j}$.
\end{proof}

Next we show that the polynomials corresponding to demand graphs with distinct values under $\mincut$ are linearly independent.

\begin{lemma}
Reorder the demand graphs such that $\mincut(K_{A_1,B_1})< \ldots < \mincut(K_{A_q,B_q})$. Then the polynomials $P_{1}, \ldots, P_{q}$ are linearly independent.
\label{lem:independence}
\end{lemma}
\begin{proof}
Let $\Pi_1, \ldots, \Pi_q$ be the optimal partitions for the instances corresponding to the demand graphs $K_{A_1,B_1}, \ldots, K_{A_q,B_q}$ respectively, i.e., for each $i\in [q]$ we have that $\mincut(K_{A_i,B_i})=\cut(\Pi_{i})$. Since $\mincut(K_{A_i,B_i})< \mincut(K_{A_j,B_j})$ whenever $i<j$, it follows that $\Pi_i$ is not feasible for the demand graph $K_{A_j,B_j}$ for all $i<j$.

Suppose that the polynomials $P_1, P_2, \ldots, P_q$ are not linearly independent. Then there exist constants $\lambda_1, \ldots, \lambda_q \in \mathbb{R}$ which are not all zero such that $P=\sum_{j \in [q]}{\lambda_j P_j}$ is the zero polynomial. We will now show that each of the constants $\lambda_1, \lambda_2, \ldots, \lambda_q$ is zero, leading to a contradiction. Instantiate $P$ on $\Pi_1$. Recall that $\Pi_1$ is not feasible for any $K_{A_i,B_i}$ with $i\geq 2$. Therefore, by Lemma~\ref{lem:feasibleNonZero}, we have that $P_{i}(\Pi_1)=0$ for all $i\geq 2$. Therefore $\lambda_{1}P_{1}(\Pi_1)=0$. Since $\Pi_1$ is an (optimal) feasible partition for instance corresponding to $K_{A_1,B_1}$, applying Lemma~\ref{lem:feasibleNonZero} we get that $P_{1}(\Pi_1)\neq 0$. This implies $\lambda_1=0$. Hence, we have $P=\sum_{2\leq j \leq q}{\lambda_j P_j}$ is the zero polynomial. Now instantiate $P$ on $\Pi_2$ to obtain $\lambda_2=0$ via a similar argument as above. In the last step, we will get that $\lambda_{q-1}=0$ and hence $P=\lambda_{q}P_q$ is the zero polynomial. Instantiating on $\Pi_{q}$ gives $0=P(\Pi_{q})= \lambda_{q}P_{q}(\Pi_q)$. Since $\Pi_{q}$ is (optimal) feasible partition for the demand graph $K_{A_q,B_q}$ it follows that $P_{q}(\Pi_q)\neq 0$, and hence $\lambda_q = 0$.
\end{proof}

Note that each of the polynomials $P_1, P_2, \ldots, P_q$ is contained in the vector space of polynomials with $n$ variables and degree $\le \alpha+\beta-1$. This vector space is spanned by $\{\prod_{v \in V}{\phi_v^{r_v}} : \; \sum_{v \in V}{r_v} \le \alpha+\beta-1\}$ and therefore is of dimension $\binom{n+(\alpha+\beta-1)}{\alpha+\beta-1}=O_{\alpha,\beta}(n^{\alpha+\beta-1})$. From Lemma~\ref{lem:independence} and the fact that size of any set of linearly independent elements is at most the size of a basis, it follows that $\Big|\distinct\{ \mincut(K_{A,B}): |A|=\alpha, |B|=\beta\}\Big|= O_{\alpha,\beta}(n^{\alpha+\beta-1})$.

\subsection{Proof via Matrices}
\label{subsection:group-upper-matrix}

In this section we show the (slightly stronger) bound that $|\distinct\{ \mincut(K_{A,B}): |A|\leq \alpha, |B|\leq \beta\}|= O_{\alpha,\beta}(n^{\alpha+\beta-1})$ using matrices. Let $\pa_{2}(V)\subseteq \pa(V)$ be the set of partitions of $V$ into exactly two parts. Let $\calQ \eqdef \{(A,B) : |A| \le \alpha, \;\; |B| \le \beta\}$.
Consider the matrix $\calM$ over $\mathbb{F}_2$  with $|\calQ|$ rows (one for each element from $\calQ$) and $|\pa_{2}(V)|=2^n$ columns (one for each partition $\Pi$ of $V$ into two parts). We now define the entries of $\cal M$. 
Given $(A,B) \in \calQ$ and $\Pi \in \pa_{2}(V)$, we set $\calM_{(A,B),\Pi}=1$ if and only if the partition $\Pi \in \pa_{2}(V)$ agrees with the demand graph $K_{A,B}$, which is equivalent to saying that $\Pi(u)\neq \Pi(v)$ whenever $u \in A$ and $v \in B$ or vice versa.

Fix a vertex $\one \in V$, and consider the set ${\cal R} \eqdef \{ (A,B) \in \calQ : \one \in A \cup B   \} \;.$

\begin{claim}
\label{lem:row-space-2-way}
Over $\mathbb{F}_2$, the row space of $\calM$ is spanned by the rows corresponding to elements from ${\cal R}$
\end{claim}
\begin{proof}
Consider $(A,B)\in \calQ$ and $\Pi\in \pa_{2}(V)$. If $\one \in A \cup B$ then $(A,B) \in {\cal R}$. Henceforth we assume that $\one\notin A\cup B$. Let
$$L(\Pi) \eqdef \calM_{(A,B),\Pi} + \sum_{A'\subset A}{\calM_{(\one \cup A',B),\Pi}} + \sum_{B'\subset B}{\calM_{(A,B'\cup \one),\Pi}} \; ,$$
where addition is over $\mathbb{F}_2$. Note that $(\one \cup A',B), (A,B'\cup \one) \in {\cal R}$ for every $A'\subset A$ and $B'\subset B$, and therefore it is enough to show that $L(\Pi) \equiv 0 \pmod{2}$.

Assume first that $\calM_{(A,B),\Pi} =1$, i.e. $\Pi$ agrees with the demand graph $K_{A,B}$.
Without loss of generality assume that $\Pi(\one)=\Pi(a)$ for some $a \in A$. Then we have $\calM_{(\one\cup A',B),\Pi}=1$ for all $A'\subset A$, and $\calM_{(A,\one\cup B'),\Pi}=0$ for all $B'\subset B$. So, $L(\Pi)=1+(2^{|A|}-1) \equiv  0 \pmod{2}$.

Otherwise, we have $\calM_{(A,B),\Pi} =0$. If for every $v \in A \cup B$ it holds that $\Pi(v) \ne \Pi(\one)$ then
$$L(\Pi) =  0 + \calM_{(\one,B),\Pi} + \calM_{(A,\one),\Pi} = 1+1 \equiv 0 \pmod{2}\; .$$
Hence suppose that there exists $v \in A \cup B$ such that $\Pi(v)=\Pi(\one)$. Without loss of generality, assume $v \in A$. Then $\calM_{(A,B'+\one),\Pi}= 0$ for all $B' \subset B$. Note that if $A_1,A_2 \subset A$ satisfy $\calM_{(\one \cup A_1,B),\Pi}=1=\calM_{(\one \cup A_2,B),\Pi}$, then $\calM_{(\one \cup A_1\cup A_2,B),\Pi} = 1$. Hence there is an inclusion-wise maximal set $A^* \subset A$ such that $\calM_{(\one \cup A^*,B),\Pi}=1$. Since $\calM_{(A,B),\Pi} =0$, we conclude that $A^* \subset A$. Moreover $|A^*|\geq 1$ since $v\in A$. Therefore
$$L(\Pi) = \calM_{(A,B),\Pi} + \sum_{A'\subset A} \calM_{(\one\cup A',B),\Pi} = \sum_{A'\subseteq A^*} \calM_{(\one\cup A',B),\Pi} = 2^{|A^*|} \equiv 0\ \text{mod}(2) \;$$
\end{proof}

An argument similar to Lemma~\ref{lem:independence} shows that rows corresponding to demand graphs with distinct values under $\mincut$ are linearly independent. Hence, we have $\Big|\distinct\{ \mincut(K_{A,B}): |A|\leq \alpha, |B|\leq \beta\}\Big|\leq \texttt{rank}(\calM)\leq |{\cal R}|$, where the last inequality follows from Claim~\ref{lem:row-space-2-way}. We now obtain the final bound
\begin{align*}
  |{\cal R}|
  &= \sum_{i\leq \alpha-1, j\leq \beta} \tbinom{n-1}{i}\cdot \tbinom{n-i-1}{j} +
  \sum_{j\leq \beta-1, i\leq \alpha} \tbinom{n-1}{j}\cdot \tbinom{n-j-1}{i} \\
 &= \sum_{i\leq \alpha-1, j\leq \beta} O_{i,j}(n^{i+j}) + \sum_{j\leq \beta-1, i\leq \alpha} O_{i,j}(n^{i+j})\\
 &= O_{\alpha,\beta} (n^{\alpha+\beta-1})
\end{align*}

\subsection{Lower Bound on Number of Distinct Cuts for \gcutk}
\label{subsection:group-lower}

We now turn to prove that the bound given in Theorem~\ref{thm:group} is existentially tight. To this end, we construct an infinite family $G_n^{\alpha,\beta}$ of graphs satisfying $|\{ \mincut(K_{A,B}): |A|=\alpha, |B|=\beta\}|\geq  \Omega_{\alpha,\beta}(n^{\alpha+\beta-1})$.

Let $n, \alpha,\beta \in \mathbb{N}$ be such that $n$ is odd, and both $\alpha$ and $\beta-1$ divide $(n-3)/2$.
We define a graph $G^{\alpha,\beta}_n$ on $n$ vertices as follows. $G^{\alpha,\beta}_n$ is composed of two graphs that share a common vertex $H^{\alpha}_n$ and $J^{\beta}_n$ defined below.
\begin{itemize}
\compactify
\item $H^{\alpha}_n$ has $(n+1)/2$ vertices, and is given by $\alpha$ parallel paths $P_1, \ldots, P_\alpha$ between two designated vertices $s,t$, each path having $(n-3)/2\alpha$ internal vertices. The edge weights are given by distinct powers of $2$, monotonically decreasing from $s$ to $t$.  All edges in $H^{\alpha}_n$ incident on $t$ have $\infty$ weight (see Figure~\ref{fig:gcut}).
\item $J^{\beta}_n$ has $(n+1)/2$ vertices, and is given by $(\beta-1)$ parallel paths $Q_1, \ldots, Q_{\beta-1}$, between $t$ and a designated vertex $u$, each having $(n-3)/2(\beta-1)$ internal vertices. As in $H^\alpha_n$, edge weights are given by distinct powers of $2$, monotonically decreasing from $t$ to $u$, and all of which are strictly smaller than the weights of $H^\alpha_n$. All edges in $J^{\beta}_n$ incident on $u$ have $\infty$ weight.

\end{itemize}

 \begin{center}
\begin{figure}[ht]

 \vspace{-5mm}
 \includegraphics[width=6.5in]{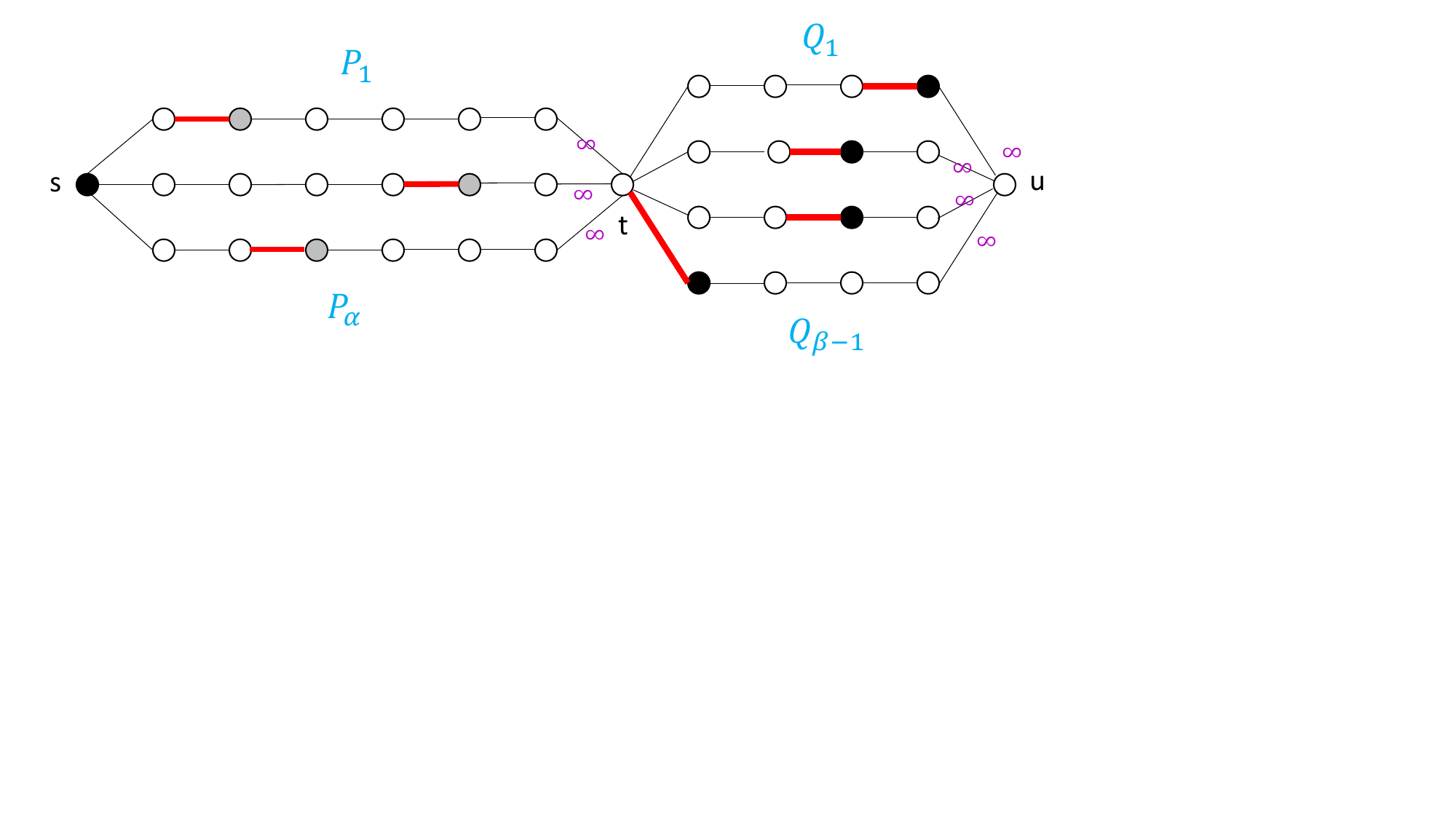}
 \vspace{-60mm}
 \caption{The graph $G^{\alpha,\beta}_n$ used in the lower bound of Section~\ref{subsection:group-lower}.
The left part of the graph is $H^{\alpha}_n$, consisting of $\alpha$ parallel $s \dash t$ paths.
The right part of the graph is $J^{\beta}_n$, consisting of $(\beta-1)$ parallel $t\dash u$ paths. The gray vertices are in $A$, and the black ones are in $B$. The red edges represent the minimum cut for this choice of $A$ and $B$.
 \label{fig:gcut}}
 \end{figure}
 \end{center}

The following claim implies the desired lower bound.
\begin{claim}
$|\{ \mincut(K_{A,B}): |A|=\alpha, |B|=\beta\}|\geq  \Omega_{\alpha,\beta}(n^{\alpha+\beta-1})$.
\end{claim}
\begin{proof}
Pick one internal vertex from each $P_i$ for $i\in [\alpha]$ to form $A$. Similarly for $\beta-1$ elements in $B$, we pick one internal vertex from each $Q_j$ for $j\in [\beta]$. In addition, $s \in B$ (as demonstrated in Figure~\ref{fig:gcut}). We claim that every such choice of $A, B$ gives a distinct value for the minimum $(A,B)$-cut.

Indeed, for $i \in [\alpha]$ let $a_i$ be the unique element in $A \cap P_i$. In order to separate $A$ from $B$, we need to separate $a_i$ from $s$. This implies that at least one edge on the segment of $P_i$ between $s$ and $a_i$ has to be in the cut. By monotonicity of weights and minimality of the cut, this must be the edge incident to $a_i$. Similarly, for every $b \in B \setminus \{s\}$, the left edge incident to $b$ must be cut. It can be easily seen (as demonstrated in Figure~\ref{fig:gcut}) that this set of edges is also enough to separate $A$ and $B$.

By the choice of weights, each such cut has a unique value, and therefore $|\{ \mincut(K_{A,B}): |A|=\alpha, |B|=\beta\}|\geq  ((n-3)/2\alpha)^{\alpha}((n-3)/2(\beta-1))^{\beta-1} = \Omega_{\alpha,\beta}(n^{\alpha+\beta-1})$.
\end{proof}

\section{\mwaycut: The Case of Clique Demands} \label{sec:multiway}

This section is devoted to proving Theorem~\ref{thm-multiway}.
In Section~\ref{subsec:multiway-upper} we show that for every graph $G=(V,E,w)$ we have $|\distinct \{ \mincut(K_S): |S|=k\}|= O_{k}(n^{k-1})$. The proof follows the lines of the proof from Section~\ref{subsection:group-upper-matrix}.
In Section~\ref{subsec:multiway-lower} we construct an infinite family of graphs for which this bound is tight. Since $|\{ K_{S}: |S|=k\}|= \binom{n}{k}= \Theta_{k}(n^{k})$, it follows that the redundancy factor is $\Omega_{k}(n)$.

\subsection{Upper Bound on Number of Distinct Cuts for \mwaycutk}
\label{subsec:multiway-upper}

In this section we show that $|\distinct \{ \mincut(K_S): |S|=k\}|= O_{k}(n^{k-1})$. Let $\pa_{k}(V)\subseteq \pa(V)$ be the set of partitions of $V$ into exactly $k$ parts. Let $\calQ \eqdef \{A\subseteq V : |A| =k\}$. Consider the matrix $\calM$ over $\mathbb{F}_2$  with $|\calQ|$ rows (one for each element from $\calQ$) and $|\pa_{k}(V)|$ columns (one for each partition $\Pi$ of $V$ into $k$ parts). We now define the entries of $\cal M$.
Given $A \in \calQ$ and $\Pi \in \pa_{k}(V)$, we set $\calM_{A,\Pi}=1$ if and only if the partition $\Pi \in \pa_{k}(V)$ agrees with the demand graph $K_A$. That is if and only if we have $\Pi(u)\neq \Pi(v)$ for every $u,v \in A$ such that $u\neq v$
Fix a vertex $\one \in V$, and consider the set ${\cal R} \eqdef \{ A \in \calQ : \one \in A  \} \;.$

\begin{claim}
\label{lem:row-space-k-way}
Over $\mathbb{F}_2$, the row space of $\calM$ is spanned by the rows corresponding to elements from ${\cal R}$
\end{claim}
\begin{proof}
Consider $A\in \calQ$ and $\Pi\in \pa_{k}(V)$. If $\one \in A $ then $A \in {\cal R}$. Henceforth we assume that $\one\notin A$. Let
$$L(\Pi) \eqdef \calM_{A,\Pi} + \sum_{a \in A} \calM_{\{\one\} \cup A \setminus \{a\},\Pi} $$
where addition is over $\mathbb{F}_2$. Note that ($\{\one\} \cup A \setminus \{a\}) \in \cal{R}$ for every $a \in A$, and hence it is enough to show that $L(\Pi) \equiv 0 \pmod{2}$.

Assume first that $\calM_{A,\Pi}=1$. Since $\Pi \in \pa_k(V)$, there is a unique element $a^* \in A$ such that $\Pi(\one)=\Pi(a^*)$. Then $\calM_{\{\one\} \cup A \setminus \{a\},\Pi}=0$ for all $a \in A \setminus \{a^*\}$ and $\calM_{\{\one\} \cup A \setminus \{a^*\},\Pi}=1$. Therefore
$$L(\Pi) \eqdef \calM_{A,\Pi} + \sum_{a \in A} \calM_{\{\one\} \cup A \setminus \{a\},\Pi} = 1 + 1 \equiv 0 \pmod{2} \;.$$
Next, assume $\calM_{A,\Pi}=0$. If $\calM_{\{\one\} \cup A \setminus \{a\},\Pi}=0$ for all $a \in A$, then clearly $L(\Pi)\equiv 0 \pmod{2}$. Otherwise, there exists $a' \in A$ such that $\calM_{\{\one\} \cup A \setminus \{a'\},\Pi}=1$. Since $\calM_{A,\Pi}=0$, it follows that $\Pi(a') \ne \Pi(\one)$. Therefore there is some $a'' \in A$ such that $\Pi(a')=\Pi(a'')$. Since $\calM_{\{\one\} \cup A \setminus \{a'\},\Pi}=1$, it follows that $\calM_{\{\one\} \cup A \setminus \{a''\},\Pi}=1$ and $\calM_{\{\one\} \cup A \setminus \{a\},\Pi}=0$ for any $a\in A\setminus \{a',a''\}$.
Therefore
$$L(\Pi) \eqdef \calM_{A,\Pi} + \sum_{a \in A} M_{\{\one\} \cup A \setminus \{a\},\Pi} = 0 + 1 + 1 \equiv 0 \pmod{2} \;.$$
\end{proof}

An argument similar to Lemma~\ref{lem:independence} shows that rows corresponding to demand graphs with distinct values under $\mincut$ are linearly independent. Hence, we have $\Big|\distinct \{ \mincut(K_S): |S|=k\}\Big|\leq \texttt{rank}(\calM)\leq |{\cal R}|$, where the last inequality follows from Claim~\ref{lem:row-space-k-way}. We now obtain the final bound since $|{\cal R}| = \binom{n-1}{k-1} = O_{k}(n^{k-1})$

\subsection{Lower Bound on Number of Distinct Cuts for \mwaycutk}
\label{subsec:multiway-lower}

We now turn to prove that the bound given in Theorem~\ref{thm-multiway} is existentially tight. To this end, we construct an infinite family $P_n$ of graphs satisfying $|\{ \mincut(K_{S}): |S|=k\}|\geq  \Omega_{k}(n^{k-1})$.

For $n\in \mathbb{N}$ consider the path graph $P_n=(V_n,E_n,w)$ where $V_n=\{1, 2, \ldots, n\}$ and $E_n=\{\{i,i+1\} : 1\leq i\leq n-1\}$. For each $i\in [n-1]$ we denote the edge $\{i,i+1\}$ by $e_i$ and set $w(e_i)=2^i$.
By choice of the weights, it follows that any set of $k-1$ edges from $E_n$ has different weight. Now consider a set $E^*\subseteq E_n$ of exactly $k-1$ edges. We will show that there is an set $S^*\subseteq V_n$ of size $k$ such that $E^*$ is the minimum solution for the \mwaycutk instance with $S^*$ as the input. Let $E^* = \{i_1, i_2, \ldots, i_{k-1}\}$. Then it is easy to see (by choice of weights) that $E^*$ is the minimum weight solution for the instance with input $S^*=\{i_1, i_2 \ldots, i_{k-1}, {i_{k-1}}+1\}$. Note that $e_{i_{k-1}}$ is an edge implies $i_{k-1} \leq n-1$ and so $i_{k-1} + 1$ is well-defined.

This implies that for the path graph $P_n$ (with the weight function specified above) we have $\Big|\distinct\{ \mincut(K_S): |S|=k\}\Big|\geq \binom{n-1}{k-1}$, and hence we have $\Big|\distinct\{ \mincut(K_S): |S|=k\}\Big|=\Omega_{k}(n^{k-1})$.

\section{\mcut: The Case of Demands with Fixed Number of Edges} \label{sec:multicut}

This section is devoted to proving Theorem~\ref{thm-multicut}.
In Section~\ref{subsec:multicut-upper} we show that for every graph $G$, $|\distinct\{ \mincut(D): D\subseteq V\times V, |D|=k\}|= O_{k}(n^{k})$. This proof follows the lines of the proof from Section~\ref{subsection:group-upper-poly}.
Then in Section~\ref{subsec:multicut-lower} we construct an infinite family of graphs for which this bound is tight. Since $|\{ D: D\subseteq V\times V, |D|=k\}|= \binom{{\binom{n}{2}}}{k}= \Theta_{k}(n^{2k})$, it follows that the redundancy factor is $\Omega_{k}(n^k)$.

\subsection{Upper Bound on Number of Distinct Cuts for \mcutk}
\label{subsec:multicut-upper}

In this section we show that $|\distinct\{ \mincut(D): D\subseteq V\times V, |D|=k\}|= O_{k}(n^{k})$. Let $r = \binom{\binom{n}{2}}{k}$ and the set of demand graphs for \mcutk be $\{ D_1, \ldots, D_r \}$. For every vertex $v \in V$ we assign a variable denoted by $\phi_v$ which can take values from $[n]$. Fix some $j\in [r]$.
Recall that $\Pi\in \pa(V)$ agrees with (or equivalently, is feasible for) the demand graph $D_j$ if and only if $u-v\in D_j$ implies $\Pi(u)\neq \Pi(v)$.

We associate with the demand graph $D_j$ the formal polynomial
$$P_j =\Pi_{u-v \in D_j}\Big(\phi_u-\phi_v\Big)$$
Note that $P_j$ is a polynomial of degree $k$.
We denote by $P_{j}(\Pi)$ the value of the polynomial $P_j$ (over $\mathbb{F}_2$) when instantiated on $\Pi$. The proof of the next lemma is straightforward.

\begin{lemma}
A partition $\Pi$ is feasible for the instance corresponding to the demand graph $D_j$ if and only if $P_{j}(\Pi)\neq 0$
\label{lem:feasibleNonZero-multicut}
\end{lemma}
\begin{proof}
Suppose $\Pi$ is feasible for the instance corresponding to the demand graph $D_j$. So for every edge $u-v\in D_j$ we have $\Pi(u)\neq \Pi(v)$ and hence $P_{j}(\Pi)\neq 0$.

Conversely, assume $P_{j}(\Pi)\neq 0$. Hence, for each edge $u-v\in D_j$ we have $\Pi(u)\neq \Pi(v)$ which is exactly the condition for $\Pi$ being feasible for the demand graph $D_j$.
\end{proof}

The proof of the following lemma is very similar to that of Lemma~\ref{lem:independence}, and hence we omit the details.

\begin{lemma}
Reorder the demand graphs such that $\mincut(D_1)<\mincut(D_2)<\ldots<\mincut(D_q)$. Then the polynomials $P_{1}, P_{2}, \ldots, P_{q}$ are linearly independent.
\label{lem:independenceMulticut}
\end{lemma}

Note that each of the polynomials $P_1, P_2, \ldots, P_q$ is contained in the vector space of polynomials with $n$ variables and degree $\leq k$. It is well known that the size of a basis of this vector space is $\binom{n+k}{k}= O_{k}(n^k)$. From Lemma~\ref{lem:independenceMulticut} and the fact that size of any set of linearly independent elements is at most the size of a basis, it follows that $\Big|\distinct\{ \mincut(D): D\subseteq V\times V, |D|=k\}\Big|= O_{k}(n^{k})$.

\subsection{Lower Bound on Number of Distinct Cuts for \mcutk}
\label{subsec:multicut-lower}

We now turn to prove that the bound given in Theorem~\ref{thm-multicut} is existentially tight. To this end, we construct an infinite family $PM_n$ of graphs satisfying $|\{ \mincut(D): D\subseteq V\times V, |D|=k\}|\geq  \Omega_{k}(n^{k})$.

For \emph{even} $n\in \mathbb{N}$  consider the graph $PM_n=(V_n,D_n,w)$ which is a perfect matching on $n$ vertices. Let $V_n=\{1, 2, \ldots, n\}$ and $D_n=\{\{2i-1,2i\} : 1\leq i\leq n/2\}$. For each $i\in [n]$ we denote the edge $\{2i-1,2i\}$ by $d_i$ and set $w(d_i)=2^i$.
By choice of the weights, it follows that any set of $k$ edges from $D_n$ has different total weight. Now consider a set $D^*\subseteq D_n$ of exactly $k$ edges. We will now show that there is an set $D^{**}\subseteq D_n$ of size $k$ such that $D^*$ is the minimum solution for the \mcutk instance whose demand graph is $D^{**}$. It is easy to see that $D^{*}$ is the only solution (and hence of minimum weight too) for the instance whose demand graph is $D^{*}$. Hence taking $D^{**}=D^*$ suffices.

This implies that for the perfect matching graph $PM_n=(V_n, D_n)$ (with the weight function specified above) we have $|\{ \mincut(D): D\subseteq V_n \times V_n, |D|=k\}|\geq \binom{n/2}{k}=\Omega_{k}(n^{k})$.

\section{Evaluation Schemes: Constructing Succinct Data Structures } \label{sec:dataStructures}

Gomory and Hu~\cite{GH61} showed that for every undirected edge-weighted graph
$G=(V,E,w)$ there is a tree $\mathcal T = (V,E',w')$ that represents
the minimum $st$-cuts exactly both in terms of the \emph{cut-values}
and in terms of their \emph{vertex-partitions}.
The common terminology for the first property,
probably due to \Benczur~\cite{benczur},
is to say that $\mathcal T$ is flow-equivalent to $G$.
The second property, which is actually stronger,
says that $\mathcal T$ is cut-equivalent to $G$.
\footnote{
We say that $\mathcal T$ is flow-equivalent to $G$ when for every $s,t\in V$
the minimum $st$-cut value in $\mathcal T$ is exactly the same as in $G$.
We say it is cut-equivalent to $G$
when every vertex-partitioning that attains a minimum $st$-cut
in $\mathcal T$, also attains a minimum $st$-cut in $G$.
}

These (flow-equivalent and cut-equivalent) trees can be viewed more generally
as succinct data structures that support certain queries,
either for the value of an optimal cut, or for its vertex-partition.

Motivated by this view, we define two types of data structures, which we call
a \emph{flow-evaluation scheme} and a \emph{cut-evaluation scheme}
(analogously to the common terminology in the literature).
These schemes are arbitrary data structures (e.g., need not form a tree),
and address the terminals version (of some cut problem).
Both of these schemes, first preprocess an input that consists
of a graph $G=(V,E,w)$, a terminals set $T\subset V$,
and a collection of demand graphs $\mathcal D$.
The preprocessed data can then be used (without further access to $G$)
to answer a cut query given by a demand graph $D\in D$.
The answer of a \emph{flow-evaluation scheme} is the corresponding
minimum cut-value $\mincut(D)$.
 The answer of a \emph{cut-evaluation scheme} is a vertex-partition
that attains this cut-value $\mincut(D)$. Formally, we define the following.

\begin{definition} \label{def:scheme}
A {\em flow-evaluation scheme} is a data structure that supports the following two operations.
\begin{enumerate}
	\item Preprocessing $P$, which gets as input a graph $G=(V,E,w)$, a set of terminals $T \subseteq V$ and a family ${\cal D}$ of demand graphs on $T$ and constructs a data structure $P(G,T,{\cal D})$.
	\item Query $Q$, which gets as input $D \in {\cal D}$ and uses $P(G,T,{\cal D})$ to output $\mincut(D)$. Note that $Q$ has no access to $G$ itself. \label{item:cutScheme}
\end{enumerate}
A {\em cut-evaluation scheme} also supports a third operation.
\begin{enumerate}
\setcounter{enumi}{2}
	\item Query $Q'$, which gets as input $D \in {\cal D}$ and uses $P(G,T,{\cal D})$ to output a partition $\Pi \in \pa(T)$ which attains $\mincut(D)$.
\end{enumerate}
\end{definition}

We provide below some constructions and lower bounds for flow-evaluation schemes and cut-equivalent schemes, for the three cut problems studied in this paper, viz. \gcut, \mwaycut and \mcut.
Note that all our upper bounds are for the stronger version of cut-evaluation schemes, and our lower bound is for the weaker version of flow-evaluation schemes for the \gcutlower problem. In order to measure bit complexity for the bounds, we assume hereafter that all weights are integers.

\subsection{Upper Bounds for Cut-Evaluation Schemes}

The next theorem follows from the \emph{terminal version} of Theorem~\ref{thm:group}. Similar results also hold for the \mwaycut and \mcut problems; the proofs follow in the same manner from Theorem~\ref{thm-multiway} and Theorem~\ref{thm-multicut} respectively.

\begin{theorem}
\label{thm:groupTerminalDS}
There exists a cut-evaluations scheme such that for every graph $G=(V,E,w)$, a set of terminals $T \subseteq V$ and $\alpha,\beta\in \mathbb{N}$, for the family ${\cal D} = \aset{ K_{A,B}: A,B \subseteq T, \; |A|=\alpha, \; |B|=\beta }$ of demand graphs at most $O_{\alpha,\beta}(|T|^{\alpha+\beta-1} \cdot (|T| + \log W))$ bits are stored, where $W = \sum_{e \in E}{w(e)}$, and such that the query time is $O_{\alpha,\beta}(|T|^{\alpha+\beta-1})$.
\end{theorem}

\begin{proof}
Let $q$ be the number of distinct values attained by demand graphs on $T$. Applying the upper bound of Theorem~\ref{thm:group} adjusted for the terminals version, we get that $q \le O_{\alpha,\beta}(|T|^{\alpha+\beta-1})$.
Order the demand graphs $\{K_{A_1,B_1}, \ldots, K_{A_q,B_q}\}$ such that $\mincut(K_{A_i,B_i}) < \mincut(K_{A_j,B_j})$ for all $i<j$. For every $j \in [q]$ we associate with $K_{A_j,B_j}$ a partition $\Pi_j \in \pa_2(T)$ which attains $\mincut(K_{A_j,B_j})$. Representing $\Pi_j$ as a bit vector of length $|T|$, we list all values in an increasing order. Each entry of this structure is of size $O(|T| + \log W)$.  Therefore the size of the evaluation scheme is at most $O_{\alpha,\beta}(|T|^{\alpha+\beta-1} \cdot (|T| + \log W))$ bits.
To see the bound of the query time, note that given $A,B \subseteq T$ such that $|A| = \alpha$ and $|B|=\beta$, the evaluation scheme holds $\mincut(K_{A,B})$ and a partition $\Pi$ that attains it. Moreover, going over the list, $\mincut(K_{A,B})$ is the first value in the list for which $\Pi$ agrees with the associated partition.
\end{proof}

\subsection{Lower Bound on Flow-Evaluation Schemes for \gcutlower}

Next we use an information-theoretic argument which shows a lower bound
on the storage required by any flow-evaluation scheme for \gcutlower.
Since a cut-evaluation scheme is stronger than a flow-evaluation scheme,
this lower bound immediately extends also to cut-evaluation schemes.

\begin{theorem}\label{th:LBMain}
For every $n \ge 3$, a flow-evaluation scheme for \gcutlower on graphs with $n$ terminals (in which $T=V$) and with edge-weights bounded by a polynomial in $n$ requires storage of $\Omega(n^2 \log n)$ bits.
\end{theorem}

Let $3 \le n \in \mathbb{N}$, let ${\cal D} = \{K_{A,B} : \; A,B \subset [n], \; |A|=2,\; |B|=1\}$ and let $G=(V,E)$ be the complete graph on $V=T=[n]$.
For every $j \in [n-1]$, $(j,j+1) \in E$ are referred to as {\em path edges}, and the rest of the edges are referred to as {\em fork} edges as demonstrated in Figure~\ref{fig:pathFork}.

\begin{figure}[ht]
  \begin{center}
  \includegraphics[scale=0.80]{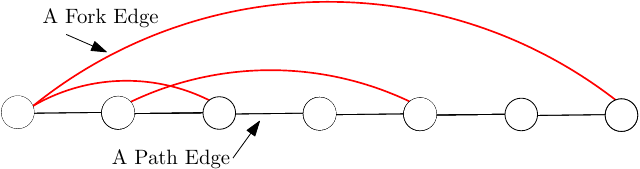}
    \caption{Edge types in $G$. Black edges correspond to {\em path} edges, while red edges correspond to {\em fork} edges (we illustrate only a few fork edges for simplicity).}
\label{fig:pathFork}
  \end{center}
\end{figure}
To prove Theorem~\ref{th:LBMain} we assign random edge weights to the graph in the range $\{1,\ldots,2n^5\}$. We then show that given query access to $\{\mincut(D) : D \in {\cal D}\}$ we can recover all edge weights. This, in turn, implies that we can recover at least $\Omega(n^2 \log n)$ bits, and thus implies Theorem~\ref{th:LBMain}.

We assign edge weights to the edges of $G$ as follows:

\begin{enumerate}[(P1)]
\item For every $j \in [n]$, $w(j,j+1)$ is chosen uniformly at random in $[2(n-j)n^4, (2(n-j)+1)n^4]$. This ensures that the weights of the path edges are non-increasing as we go from left to right. Moreover, whenever $j>i$,
$$w(j,j+1)\leq (2(n-j)+1)n^4 = (2n-2j+1)n^4< (2n-2i)n^4 - n^4 \leq w(i,i+1) - n^4\;.$$

\item The weights of all fork edges are chosen uniformly at random from $\{0,1,2,\ldots,n-1\}$. Thus the total weight to all fork edges is at most $n^3$. Note that this is strictly smaller than the difference between weights of any two path edges (which is at least $n^4$).
\end{enumerate}

\begin{claim} \label{c:cut}
\ifprocs
Let $1 \le i < j \le n$, then $\mincut(K_{\{1,j\},\{i\}}) = \sum\limits_{e \in E : | e \cap \{i,\ldots,j-1\}|=1}{w_e}$.
\else
Let $1 \le i < j \le n$, then $\mincut(K_{\{1,j\},\{i\}}) =\sum\limits_{e \in E : | e \cap \{i,\ldots,j-1\}|=1}{w_e}$.
\fi
\end{claim}
\begin{proof}
By Property (P2) the total weigtht of all the fork edges is at most $n^3$, which is less than the difference between weights of any two path edges (which is at least $n^4$). Hence, the value $\mincut(K_{\{1,j\},\{i\}})$ is determined only by which path edges we choose.

Note that the minimum cut separating $\{1,j\}$ from $\{i\}$ needs to pick at least one edge each from the paths $1-i$ and $i-j$. By Property (P1) the path edges have non-increasing weights going from left to right, and hence the two cheapest edges on the paths $1-2-\ldots-i$ and $i-(i+1)-\ldots-j$ are $(i-1,i)$ and $(j-1,j)$ respectively. Therefore, it follows that $\mincut(K_{\{1,j\},\{i\}}) = \cut(C,\overline{C})$, where $C=\{i,i+1, \ldots, j-1\}$.
By definition of $\cut(C,\overline{C})$ it follows that $$\mincut(K_{\{1,j\},\{i\}}) = \sum\limits_{e \in E : | e \cap \{i,\ldots,j-1\}|=1}{w_e}$$
%
%
%
\end{proof}

\begin{lemma}\label{l:recover}
Given access to queries $Q$ as in Definition~\ref{def:scheme}, we can recover $w$.
\end{lemma}
Lemma~\ref{l:recover} implies that we can recover $\Omega(n^2 \log n)$ random bits given access to queries $Q$ as in Definition~\ref{def:scheme}, and thus implies Theorem~\ref{th:LBMain}. 

\begin{proof}
For every $1 \le i < j \le n$, we show that we can recover $w_{ij}$. The proof continues by induction on $j-i$, starting with the case $j=i+1$.
By Claim~\ref{c:cut}, we get that $\mincut(K_{\{1,i+1\},\{i\}}) = \sum_{e \in E : i \in e}{w_e}$ and $\mincut(K_{\{1,i+2\},\{i+1\}}) = \sum_{e \in E : i+1 \in e}{w_e}$. Therefore
\begin{equation}
\mincut(K_{\{1,i+1\},\{i\}}) + \mincut(K_{\{1,i+2\},\{i+1\}})
= 2w_{i,i+1} + \sum_{e \in E : |e \cap \{i,i+1\}| = 1}{w_e}
\label{eq:sumCuts}
\end{equation}
In addition we have $\mincut(K_{\{1,i+2\},\{i\}}) = \sum_{e \in E : |e \cap \{i,i+1\}| = 1}{w_e}$. Plugging this into \eqref{eq:sumCuts} we get that
\ifprocs
\begin{equation*}
\begin{split}
w_{i,i+1} = \frac{1}{2}\left( \right.\mincut(K(\{1,i+1\},\{i\})) + &\mincut(K(\{1,i+2\},\{i+1\})) \\
&- \left. \mincut(K(\{1,i+2\},\{i\}))\right) \;,
\end{split}
\end{equation*}
\else
$$w_{i,i+1} = \frac{1}{2}\left( \mincut(K_{\{1,i+1\},\{i\}}) + \mincut(K_{\{1,i+2\},\{i+1\}}) - \mincut(K_{\{1,i+2\},\{i\}})\right) \;,$$
\fi
and therefore we can recover $w_{i,i+1}$. Next, let $1 \le i < j \le n$, and assume that we can recover $w_{pq}$ for all $1 \le p < q \le n$ such that $q-p < j-i$.
In addition, we assume that $j<n$. The proof is similar for the case $j=n$.
From Claim~\ref{c:cut} we get that
\begin{equation}
\mincut(K_{\{1,j+1\},\{i\}})= \sum_{e \in E : | e \cap \{i,\ldots,j\}|=1}{w_e} = \sum_{k=i}^j{\sum_{m \notin \{i,\ldots,j\}}{w_{mk}}}
\label{eq:cut}
\end{equation}
Let $i < k < j$, then by Claim~\ref{c:cut}
\ifprocs
\begin{equation*}
\begin{split}
\sum_{m \notin \{i,\ldots,j\}}{w_{mk}} &= \sum_{e \in E : k \in e}{w_e} -  \sum_{m \in \{i,\ldots,j\}}{w_{mk}} \\
&= \mincut(K(\{1,k+1\},\{k\})) - \sum_{m \in \{i,\ldots,j\}}{w_{mk}} \;.
\end{split}
\end{equation*}
\else
$$\sum_{m \notin \{i,\ldots,j\}}{w_{mk}} = \sum_{e \in E : k \in e}{w_e} -  \sum_{m \in \{i,\ldots,j\}}{w_{mk}} = \mincut(K_{\{1,k+1\},\{k\}}) - \sum_{m \in \{i,\ldots,j\}}{w_{mk}} \;.$$
\fi
For every $m \in \{i,\ldots,j\}$, $|k-m| < j-i$, and therefore we can recover $w_{mk}$. It follows that we can recover $\sum_{m \notin \{i,\ldots,j\}}{w_{mk}}$.
Plugging this into \eqref{eq:cut}, and rearranging we get that
\ifprocs
\begin{equation}
\begin{split}
\mincut(K(\{1,j+1\},i)) &- \sum_{i<k<j}{\sum_{m \notin \{i,\ldots,j\}}{w_{mk}}} \\
&= \sum_{m \notin \{i,\ldots,j\}}{w_{im}} + \sum_{m \notin \{i,\ldots,j\}}{w_{jm}}
\end{split}
\label{eq:cut2}
\end{equation}
\else
\begin{equation}
\mincut(K_{\{1,j+1\},\{i\}}) - \sum_{i<k<j}{\sum_{m \notin \{i,\ldots,j\}}{w_{mk}}}= \sum_{m \notin \{i,\ldots,j\}}{w_{im}} + \sum_{m \notin \{i,\ldots,j\}}{w_{jm}}
\label{eq:cut2}
\end{equation}
\fi
It remains to show that we can recover $w_{ij}$ assuming we can recover $\sum_{m \notin \{i,\ldots,j\}}{w_{im}} + \sum_{m \notin \{i,\ldots,j\}}{w_{jm}}$. Applying Claim~\ref{c:cut} once more, we get that
\ifprocs
\begin{equation*}
\begin{split}
\sum_{m \notin \{i,\ldots,j\}}{w_{im}} &= \sum_{e \in E : i \in e}{w_e} -  \sum_{m \in \{i,\ldots,j\}}{w_{im}}\\
&=\mincut(K(\{1,i+1\},\{i\})) - w_{ij} - \sum_{m \in \{i,\ldots,j-1\}}{w_{im}} \;,
\end{split}
\end{equation*}
\else
$$\sum_{m \notin \{i,\ldots,j\}}{w_{im}} = \sum_{e \in E : i \in e}{w_e} -  \sum_{m \in \{i,\ldots,j\}}{w_{im}}=\mincut(K_{\{1,i+1\},\{i\}}) - w_{ij} - \sum_{m \in \{i,\ldots,j-1\}}{w_{im}} \;,$$
\fi

and similarly
$$\sum_{m \notin \{i,\ldots,j\}}{w_{jm}} = \mincut(K_{\{1,j+1\},\{j\}}) - w_{ij} - \sum_{m \in \{i+1,\ldots,j\}}{w_{im}} \;.$$
By the induction hypothesis, we can recover $w_{im}$ for all $m \in \{i,\ldots,j-1\}$ and $w_{jm}$ for all $m \in \{i+1,\ldots,j\}$, and therefore we can recover $w_{ij}$.
\end{proof}
This completes the proof of Theorem~\ref{th:LBMain}.
We note that similar arguments give a lower bound of $\Omega(n^3 \log n)$ by allowing weights which are exponential in $n^3$.
Details omitted.

\section{Future Directions}

A natural direction for future work is to construct better data structures
for the problems discussed in this paper.
Our tight bounds on the number of distinct cut values (redundancy factor)
yield straightforward schemes with improved storage requirement,
as described in Section~\ref{sec:dataStructures}.
But one may potentially improve these schemes in several respects.
First, our storage requirement exceeds by a factor of $|T|$
the number of distinct cut values.
The latter (number of distinct cut values)
may be the ``right bound'' for storage requirement,
and it is thus important to prove storage lower bounds;
we only proved this for \gcutlower.
Second, it would be desirable to achieve fast query time,
say sublinear in $|T|$ or perhaps even constant.
Third, one may ask for a distributed version of the data structure
(i.e., a labeling scheme) that can report the same cut values;
this would extend the known results \cite{KKKP05} for minimum $st$-cuts.
All these improvements require better understanding of the structure
of the optimal vertex partitions (those that attain minimum cut values).
Such structure is known for minimum $st$-cuts,
where the Gomory-Hu tree essentially shows the existence of
a family of minimum $st$-cuts, one for each $s,t \in V$, which is laminar.

Another very interesting question is to explore approximation
to the minimum cut, i.e., versions of the above problems where
we only seek for each instance a cut within a small factor of the optimal.
For instance, the cut values of \gcutab can be easily approximated
within factor $\alpha\cdot \beta$ using Gomory-Hu trees,
which requires storage that is linear in $|T|$,
much below the aforementioned ``right bound'' $|T|^{\alpha+\beta-1}$.
Can a better approximation be achieved using similar storage?

\bibliographystyle{alphaurlinit}
\bibliography{cuts}

\appendix

\section{No Non-trivial Redundancy for Directed Graphs}
\label{app:directed-graphs}

In this section, we consider the directed versions of the three cut problems considered in this paper, viz. the \gcut, \mwaycut, \mcut. Note that in directed graphs, an $s\rightarrow t$ cut is a set of edges whose removal ensures there is no $s\rightarrow t$ path. We construct an infinite family of graphs which have no non-trivial redundancy for any of these problems, i.e., the number of distinct cut values is asymptotically equal to the total number of instances.

Let $n\in \mathbb{N}$, and let $X,Y$ be two disjoint $n$-element sets. Consider the graph $G_n \eqdef K_{X \to Y}$, which is the orientation of the complete bipartite graph $K_{X,Y}$ obtained by orienting each edge from a vertex of $X$ towards a vertex of $Y$. We assign edge weights in $G_n$ in such a manner that every set of edges has distinct weight (for example, we may assign each edge a distinct power of $2$).


\paragraph{\gcut in Directed Graphs.}
In the directed version of the \gcutab problem, given sets $A,B \subseteq V$ such that $|A|=\alpha$ and $|B|= \beta$,
we want to find a set of edges of minimum weight whose removal ensures there is no path from any vertex of $A$ to any vertex of $B$.
The total number of demand graphs for $G_n$ is therefore
$|\{K_{A\rightarrow B}: |A|=\alpha,|B|=\beta\}|= \binom{2n}{\alpha}\cdot \binom{2n-\alpha}{\beta} = \Theta_{\alpha,\beta}(n^{\alpha+\beta})$.
Let $A \subseteq X,B \subseteq Y$ be such that $|A| = \alpha, |B|= \beta$. A minimum $(A,B)$-cut must include all the edges of $K_{A\rightarrow B}$, and furthermore these edges are enough. By choice of weights, the value of the minimum $(A,B)$-cut is unique. This implies that in $G_n$ we have $|\distinct\{ \mincut(K_{A\rightarrow B}): |A|=\alpha, |B|=\beta\}|\geq \binom{n}{\alpha} \cdot \binom{n}{\beta} = \Omega_{\alpha,\beta}(n^{\alpha+\beta})$
, and hence there is no non-trivial redundancy.

We note that for the special case of $st$-cuts in directed graphs (i.e. \gcutab with $\alpha = \beta =1$), Lacki \etal \cite{LNSW12} show that there exists an infinite family of {\em planar} graphs, which have no non-trivial redundancy. That is, for every graph in the family there are $\Theta(|V|^2)$ distinct $st$-cuts.

\paragraph{\mwaycut in Directed Graphs.}
In the directed version of the \mwaycutk problem, given a $k$-element set $S \subseteq V$ we want to find a set of edges of minimum weight whose removal ensures there is no $s\rightarrow s'$ path for any distinct $s,s'\in S$.
Let $k \le n$ be even. The number of instances $S$ in $G_n$ is $\binom{2n}{k} = \Theta_{k}(n^{k})$.
Let $A \subseteq X , B \subseteq Y$ be such that $|A|=|B|=k/2$, and let $S = A \cup B$. Then $|S|=k$, and therefore constitutes an instance for the directed \mwaycutk problem. For this instance any multiway cut must include all the edges of $K_{A\rightarrow B}$, and furthermore these edges are enough. Therefore the number of distinct cut values for the directed \mwaycut problem is at least $|\distinct\{ \mincut(K_{A\rightarrow B}): A \subseteq X, B \subseteq Y, |A|=|B|=k/2 \}| = \binom{n}{k/2}\cdot \binom{n}{k/2} = \Omega_{k}(n^{k})$.

\paragraph{\mcut in Directed Graphs.}
In the directed version of the \mcutk problem, given a set of demands $D\subseteq V\times V$ such that $|D|=k$, we want to find a set of edges of minimum weight whose removal ensures there is no $s\rightarrow s'$ path for any $(s,s')\in D$.
The total number of such demand graphs for $G_n$ is $|\{D\subseteq V\times V : |D|=k\}|= \binom{2n(2n-1)}{k} =\Theta_{k}(n^{2k})$.
Let $D$ be a set of edges of $G_n$ such that $|D| = k$. It is evident that the minimum set of edges satisfying $D$ is, in fact, $D$ itself. By the choice of weights, the number of distinct values is at least $|\distinct\{ \mincut(D): D\subseteq V\times V, |D|=k\}|\geq \binom{n^{2}}{k}= \Omega_{k}(n^{2k})$.

\end{document}